%% file: WG.tex
\documentclass[11pt]{article}
\pdfoutput=1
\usepackage[letterpaper,margin=1in]{geometry}

\usepackage{graphicx}
\usepackage{natbib}
\usepackage{amssymb}
\usepackage[usenames]{color}
\usepackage{array}     	
\usepackage{amsthm}   	
\usepackage[ruled, vlined]{algorithm2e}
\usepackage{verbatim}  	
\usepackage{tabularx}  	
\usepackage{multirow}  	
\usepackage{booktabs} 	
\usepackage[table]{xcolor}  
\usepackage{subfig}  		

\renewcommand{\multirowsetup}{\centering}

\newcommand{\dist}{{\rm dist}}
\newcommand{\near}{{\rm near}}

\newcommand{\G}{{\widehat G}}
\newcommand{\V}{{\widehat V}}
\newcommand{\E}{{\widehat E}}
\newcommand{\w}{{\widehat w}}
\newcommand{\X}{{\widehat X}}

\newtheorem{theorem}{Theorem}[section]
\newtheorem{lemma}[theorem]{Lemma}

\newtheorem{observation}[theorem]{Observation}
\newtheorem{definition}[theorem]{Definition}
\newtheorem{corollary}[theorem]{Corollary}

\begin{document}

\sloppy



\title{An Oblivious Spanning Tree for \\Buy-at-Bulk Network Design Problems}

\author{
 {Srivathsan Srinivasagopalan, Costas Busch, S.S. Iyengar} \\
   {Computer Science Department}	\\
   {Louisiana State University}	\\
   {\{ssrini1, busch, iyengar\}@csc.lsu.edu}
}

\date{}  

\maketitle

\begin{abstract}
We consider the problem of constructing a single spanning tree for the single-source buy-at-bulk network design problem for doubling-dimension graphs. We compute a spanning tree to route a set of demands (or data) along a graph to or from a designated root node. The demands could be aggregated at (or symmetrically distributed to) intermediate nodes where the fusion-cost is specified by a non-negative concave function $f$. We describe a novel approach for developing an oblivious spanning tree in the sense that it is independent of the number of data sources (or demands) and cost function  at intermediate nodes.

	\hspace{1cm}To our knowledge, this is the first paper to propose a single spanning tree solution to this problem (as opposed to multiple overlay trees). There has been no prior work where the tree is oblivious to both the fusion cost function and the set of sources (demands). We present a deterministic, polynomial-time algorithm for constructing a spanning tree in low doubling graphs that guarantees $\log^{3}D\cdot\log n$-approximation over the optimal cost, where $D$ is the diameter of the graph and $n$ the total number of nodes. With constant fusion-cost function our spanning tree gives a $O(\log^3 D)$-approximation for every Steiner tree to the root. 


\end{abstract}


\input{Introduction}

\input{definitions}

\input{SpanningTree}

\input{VirtualTree}
\input{Analysis}
\input{Simulation}
\input{Conclusion}

\bibliographystyle{plainnat}
\bibliography{Bibliography}



\appendix
\newpage
\pagenumbering{roman}
\input{Proofs}

\end{document}

%% file: Introduction.tex
\section{Introduction}

Buy-at-bulk network design problems arise in scenarios where economies of scale applies or when the availability of capacity in discrete units result in concave cost function on the edges. As observed in \cite{1170547}, a commonly seen application is in telecommunication networks where bandwidth on a link can be purchased in some discrete units $u_1 < u_2 < \dots < u_n$ with respective costs $c_1 < c_2 < \dots < c_n$. The economies of scale exhibits the property where the cost per bandwidth decreases as the number of units purchased gets larger: $c_1/u_1 > c_2/u_2 > \dots c_n/u_n$. This behavior justifies the sale of network capacity in ``wholesale'' (or ``volume discount'') where more the capacity is bought, the cheaper is the price per unit of bandwidth.

We study the single-source buy-at-bulk (SSBB) network design problem with the following constraints: an unknown number of source (or demand) nodes and unknown concave transportation cost function $f$. An abstraction of this problem can be found in many applications, one of which is data fusion in wireless sensor networks where the given constraints are assumed unknown or vary over time. Others include design of VLSI power circuitry, Transportation \& Logistics (railroad, water, oil, gas pipeline construction) etc. For simplicity, we consider data fusion problem in communication networks, though SSBB can also be applied to data distribution problems; our solution holds for both the cases. 

As mentioned in \cite{644191}, if information flows from $k$ different sources over a link, then, the total information that needs to be transmitted is $f(k)$, where the function $f$ is called a \textit{canonical} fusion function where $f$ is concave, non-decreasing and $f(0)=0$. If $f$ is known in advance, then, the problem of building an optimal fusion tree is well understood \cite{276725}. However, here we consider the case where both $f$ and the number of data sources are unknown. The focus of this paper is to construct an oblivious spanning tree according to the following definition:

\begin{definition} [\bf{$k$-Oblivious Spanning Tree ($k$-OST)}]
A spanning tree $T$ with root node $s$ in a graph $G$ is $k$-oblivious if it provides a $k$-factor approximation to any data fusion problem toward sink node $s$ with an arbitrary set of data sources and arbitrary canonical fusion function $f$.
\end{definition}

In this paper, we consider building an oblivious spanning tree for doubling dimension graphs. Doubling dimension graphs has been used in many different contexts including compact routing in wired networks \citep{1154086,1432318}, traveling salesman, navigability and problems related to modeling the structural properties of the Internet distance matrix for distance estimation \citep{1568322, Fraigniaud07theinframetric}. As noted in \cite{Fraigniaud06adoubling}, it has become a key concept to measure the ability of network to support efficient algorithms or to realize specific tasks efficiently. For wireless networks, this concept has found many uses in solving many distributed communication problems \citep{1073826}, distributed resource-management \citep{gao09distributed}, information exchange among producers and consumers \citep{FuGui2006}, and for determining other performance qualities such as energy-conservation in wireless sensor networks \citep{1132922}.

SSBB is NP-Hard as the Steiner Tree problem is its special case (when $f(x) = 1$) due to reduction from Steiner tree problem \cite{589033}.  

\subsection{Contribution}

We build a \emph{single} oblivious spanning tree for doubling-dimension graphs, such that the tree is independent of the data sources, and can accommodate \emph{any} canonical fusion-cost function. Our approach gives a deterministic, polynomial-time algorithm that guarantees $O(2^{10\rho}\log^{3}D \log n)$-OST (measured as the total involved communication cost), where $\rho$ is the doubling dimension factor. For constant fusion cost functions, the $k$-OST behaves like a $k$-oblivious steiner tree. We define a \textit{$k$-oblivious steiner tree} as a steiner tree that provides a $k$-approximation to any data fusion problem towards the sink with arbitrary number of data sources. For constant fusion-cost function $(f(\cdot) = c)$, we obtain a $O(2^{10\rho}\log^3 D)$-oblivious steiner tree to sink $s$.

Our spanning tree construction is based on the following techniques. We partition the nodes in a hierarchical fashion. The selection of nodes for a given `level' of hierarchy is based on their mutual distances proportional to the level. Nodes of successive levels are connected by shortest paths. The intersecting paths are appropriately modified to result in a spanning tree. A modified tree is built from the spanning tree to ensure that all paths have appropriate end-nodes. Analysis is done on this modified tree.

\subsection{Related Work}

SSBB problems have been primarily considered in both Operations Research and Computer Science literatures in the context of flows with concave costs. SSBB problem was first introduced by \cite{589033} and a $O(\log^2 n)$-approximation was given by  \cite{796341} using metric tree embedding. \cite{276725} further improved this result to $O(\log n)$. \cite{380827} provided the first constant-factor approximation to the problem.

\renewcommand{\multirowsetup}{\centering}
\newcommand {\otoprule}{\midrule [\heavyrulewidth]}  
\newcommand{\capbot}[1]{%
\vspace{8pt}\footnotesize\raggedright #1}%
	
\begin{table}																										
\begin{center}
	\caption{Our results and comparison with previous results for data-fusion schemes. IAF - Independent of fusion function, ISS - Independent of source-set, $n$ is the total number of nodes in the topology, $k$ is the total number of source nodes.}	
  \footnotesize
  \renewcommand{\arraystretch}{1.5}
 	\begin{tabular}{*{3}{m{2.2cm}}*{2}{>{$}c<{$}}*{1}{>{$}l<{$}}{m{1.4cm}}}  
    \otoprule
    \multicolumn{1}{b{1.5cm}}{\bfseries Related Work}
    & \multicolumn{1}{b{1.5cm}}{\bfseries Algorithm Type}
    & \multicolumn{1}{b{1.5cm}}{\bfseries Graph Type}
    & \multicolumn{1}{b{0.9cm}}{\bfseries IAF}
    & \multicolumn{1}{b{0.9cm}}{\bfseries ISS}
    & \multicolumn{1}{b{1.5cm}}{\bfseries Approx Factor}  
    & \multicolumn{1}{b{1.3cm}}{\bfseries Tree Type}\\
    \otoprule 
    \raggedright Lujun Jia\\ \textit{et al.} \cite{JiaNRS06}
    & Deterministic
    & \raggedright Random\\Deployment
    & \times
    & \checkmark
    & O(\log n)
    & One Overlay \\
    \hline
    \raggedright Ashish Goel\\ \textit{et al.} \cite{644191}
    & Randomized
    & General Graph $\bigtriangleup$-inequality
    & \checkmark
    & \times
    & O(\log k) 
    & One Overlay \\
    \hline
    \raggedright Ashish Goel\\ \textit{et al.} \cite{DBLP:journals/corr/abs-0908-3740}
    & Probabilistic
    & General Graph
    & \checkmark
    & \times
    & O(1) 
    & Multiple Overlay \\
    \hline
    \rowcolor[gray]{.95}
    This paper
    & Deterministic
    & \raggedright Low Doubling\\ Dimension
    & \checkmark
    & \checkmark
    & O(\log^3D\cdot\log n)
    & One Spanning \\
    \bottomrule
 
  \end{tabular}
  \label{table:comparison}

\end{center} 
\end{table}

\cite{644191} build an overlay tree on graphs that satisfy triangle-inequality based on maximum matching algorithm that guarantees $1 + \log k$ approximation, where $k$ is the number of sources. An overlay tree, if projected to a graph, may not be a tree (could have cycles). In a related paper by \cite{DBLP:journals/corr/abs-0908-3740}, the authors construct (in polynomial time) a set of overlay trees from a given general graph such that the expected cost of a tree for any $f$ is within an $O(1)$-factor of the optimum cost for that $f$. 

\cite{JiaNRS06} build Group Independent Spanning Tree Algorithm (GIST) that constructs an overlay tree for randomly deployed nodes in 2D. The tree (that is oblivious to the number of data sources) simultaneously achieves $O(\log n)$-approximate fusion cost and $O(1)$ approximate delay. However, their solution assumes a constant fusion cost function. We summarize and compare the related work in Table \ref{table:comparison}.

\cite{1060649} provide approximation algorithms for TSP, Steiner Tree and set cover problems. They present a polynomial-time $(O(\log(n)),O(\log(n)))$-partition scheme for general metric spaces. An improved partition scheme for doubling metric spaces is also presented that incorporates constant dimensional Euclidean spaces and growth-restricted metric spaces. The authors present a polynomial-time algorithm for Universal Steiner Tree (UST) that achieves polylogarithmic stretch with an approximation guarantee of $O(\log^4 n/\log \log(n))$ for arbitrary metrics and derive a logarithmic stretch, $O(\log(n))$ for any doubling, Euclidean, or growth-restricted metric space over $n$ vertices.

An earlier version of this work with preliminary results in the context of data aggregation appeared as a brief announcement in \cite{sri-algosensors-2009}.


%% file: definitions.tex
\section{Definitions}

Consider a weighted graph $G = (V, E, w)$, $w: E \longrightarrow \mathbb{Z}^+$. Let $s \in V$ be the sink node. For any two nodes $u,v \in V$ let $\dist_G(u,v)$ denote the \textit{distance} between $u,v$ (measured as the total weight of the shortest path that connects $u$ and $v$). Let $D$ denote the \textit{diameter} of $G$, that is, $D = \max_{u,v \in V} \dist_G(u,v)$. Given a subset $V' \subseteq V$, we denote $\dist_G(u,V')$ the smallest distance between $u$ and any node in $V'$. We also define $\near_G(u,V') = \{ v \in V' :~ \dist_G(u,v) = \dist_G(u,V') \}$ to be the set of nodes in $V'$ that $u$ is closest to. Given a set of nodes $H \subseteq V$ and parameter $d$, we define \emph{Maximal Independent Set of $G$ for distance $d$} as $I = MIS(G, H, d)$ to be an arbitrary maximal set of nodes in $G$ such that $H \subseteq I$ and for any $u,v \in G \setminus H$, $\dist_G(u,v) \geq d$ and $\dist_G(u, x) \geq d$ where $x \in H$. There is a simple greedy algorithm to compute it. A special case is the notion of $d$-independent nodes computed by $I = MIS(G,\phi,d)$.

We adapt the definition of doubling-dimension graph from \citet{eemcs8079, 946308}.

\begin{definition} [\bf{r-Neighborhood}]
Given a graph $G = (V,E)$, the $r$-neighborhood of any vertex $u \in V$, $N_{G}(u,r)$, is defined as the set of nodes whose distance (in hops) is at most $r$ from $u$; $N_{G}(u,r) = \{v~|~dist_{G}(u,v) \leq r\}$. 
\end{definition}

\begin{definition} [\bf{Doubling Dimension Graph}]
A graph $G$ with the smallest $\rho$ such that every $R$-neighborhood is completely covered by atmost $2^\rho$ $R/2$-neighborhood is said to doubling. If $\rho$ is bounded by a constant and is small, we say that $G$ is \textbf{doubling} and has a low dimension. 
\end{definition}

\begin{observation}
In any $R$-neighborhood of doubling dimension graphs, there are at most $m$ $R/(2^j)$-neighborhoods where $m = 2^{\rho \log (\frac{R}{R/(2^j)})}$, $j \geq 0$.
\end{observation}

%% file: SpanningTree.tex
\section{Spanning Tree Construction}

We start with an informal description of the construction of the spanning tree. We build the tree in a hierarchical manner that has $\kappa = O(\log D)$ levels. At level $\kappa = \lceil\log D\rceil$ is the sink ($I_\kappa = \{s\}$) and at level $\kappa = 0$ are the individual nodes ($I_0 = V$). Each level $i$ of the hierarchy is built by identifying a set of independent nodes, $I_i = \{\ell_i^1,\ell_i^1,\dots,\ell_i^n\}$, ($n\geq 0$ and the subscript corresponds to level), that are $2^i$ distance (in hops) apart. The hierarchy consists of $\kappa$ levels of independent nodes $I_0,\ldots,I_\kappa$ where $I_i \subseteq V$, $0\leq i \leq \kappa$. Members of $I_i$ are also called leaders of level $i$. Some leaders could belong to more than one level (eg., sink $s$). Leaders of consecutive levels are connected by shortest paths to build a tree. A formal description appears in Algorithm \ref{alg:SpanningTree}.


\begin{algorithm}
\KwIn{Graph $G$ with sink $s$.}
\KwOut{A spanning tree $T_s$.}
\BlankLine

$P \gets \phi$; $I_{\kappa} \gets \{s\}$ \tcp*{$\kappa \gets \lceil\log D\rceil$}
$P^{reg} \gets \phi$; $P^{pr} \gets \phi$ \tcp*{List of regular and pruned paths}

\ForEach {\mbox{level} $i = \kappa - 1 ~\KwTo~  0$} {
$I_i \gets MIS(G, I_{i+1}, 2^i)$\;

\ForEach {$v \in I_i$}{
	\tcp{$v$ chooses nearest $\ell \in I_{i+1}$ as parent.}
	$\ell \gets near_G(v, I_{i+1})$  \tcp*{If $s \in near(v, I_{i+1})$, $s$ is chosen.}
	$p_v \gets \mbox{Compute shortest path from $v$ to } \ell$\;
	\eIf {$p_v$ intersects any path at level $> i$ at point $u$}{
		\tcp{Prune path $p_v$ by removing segment from $u$ to $\ell$}
		$p^\prime \gets \mbox{path segment from $v$ to } u$\;
		$P^{pr}_i \gets P^{pr}_i \cup p^\prime$\;
	}{
		$P^{reg}_i \gets P^{reg}_i \cup p_v$\;
	}
}
}

$P \gets \bigcup_{i=0}^{i = \kappa-1}P^{reg}_i \cup \bigcup_{i=0}^{i = \kappa-1}P^{pr}_i$

\Return $T_s$ formed by the paths in $P$\;
\caption{Spanning Tree}
\label{alg:SpanningTree}
\end{algorithm}


The construction of hierarchical levels of independent nodes is top-down. $I_i$ is computed by $MIS(G,I_{i+1},2^i)$, for $0 \leq i \leq \kappa -1$. $I_i$ will contain all the $2^j$-independent nodes of higher levels $j$, $i < j \leq \kappa$ as well as a $2^i$-independent set of nodes. We enforce the constraint that $s \in I_i$ for every $I_i$. Note that each node $v \in I_i \setminus I_{i+1}$ has to be within distance $2^{i+1} - 1$ to at least one node in $I_{i+1}$ (otherwise $v$ must be a member of $I_{i+1}$).

Paths are also constructed in a top-down fashion. A path at any level $i$, $p_i$ starts at some leader at level $i$ and ends at a leader at level $i+1$. A set of all paths at level $i$ is denoted as $P_i$ and the set of all paths of all levels is denoted by $P = \{P_{\kappa-1}, P_{\kappa-2},\ldots, P_2, P_1\}$. We begin constructing paths from level $\kappa-1$ to sink $s$ by computing shortest paths from each node of $I_{\kappa-1}$ to $s$. We continue constructing the spanning tree by computing the paths for the remaining levels. Suppose we have constructed paths for level $i$; we will construct $P_{i-1}$. Each node $v \in I_{i-1}$ \emph{prefers} some arbitrary node in $near_G(v,I_i)$; if $s \in near_G(v,I_i)$, then, $s$ is preferred. Shortest paths from each node $v$ are constructed to their respective leaders.

When paths for all levels are built, the resulting stucture may not be a tree. It could result in a graph that might have intersecting paths. Define \emph{regular} paths as paths that do not intersect any (higher-level) path on their way to their end-nodes. The paths of $P_{\kappa-1}$, are regular paths, since there were no higher-level paths to intersect and are included in $P^{reg}_{\kappa-1}$. 

Define \emph{pruned paths} as those paths that intersect paths of higher level. If a path $p_{i}$ intersects a path $p_j$ ($j > i$) along its way to $\ell_{i+1}$, $p_{i}$ is pruned from the intersection point to its destination. Such paths are included in $P^{pr}_{i}$. This pruning of intersecting paths ensures the structural property of a spanning tree (see Figure \ref{fig:ModifiedPath}). 

Note that regular paths of the same level could intersect and continue on different directions to reach a common leader. In this case, one of the paths is modified to use the same segment as the other after the intersection point. Another scenario is when two paths (say from $u$ and $v$ of level $i$) intersect at $m$ and proceed to their respective endnodes $x$ and $y$.  In this case, either $v$ or $u$ will choose a common leader and appropriately modify its path. In both these scenarios, the resulting paths remain regular even if they overlap. Note that in both the cases, the path segments, after intersection, should have the same length.

The spanning tree algorithm executes in $\kappa$ rounds and each round computes $MIS$ (in $O(|E|)$ time, assuming that the input is given by an adjacency list) and shortest paths (in $|E|\log n)$ time steps). This amounts to a total running time of $O(\log D \cdot |E|\log n)$.

%% file: VirtualTree.tex
\section{Modified Tree Construction}
The pruned paths in the spanning tree $T$ will not have leaders as end-nodes. To ensure that end-nodes of all paths are leaders, we modify $T$ to $\overline{T}$. We begin with an overview of the modified tree construction. We construct $\overline{T}$ from $T$ by assigning alternate leaders to those paths whose `upper' sections have been pruned. We first begin by assigning \emph{levels} to all the nodes of regular paths by {\ttfamily{AssignLevels}} (see appendix) and including those paths in $\overline{T}$. Then, we begin a top-down, level-by-level process where we `modify' the pruned paths by extending the pruned paths to their newly assigned alternate leaders. Note that a modified path could be a concatenation of multiple pruned paths. Then, we assign levels to the nodes of the recently modified path as well and include this modified path in $\overline{T}$. The end of this process results in a tree $\overline{T}$. A more formal description appears in Algorithm \ref{alg:ModifiedTree}.


\begin{algorithm}[!ht]
\SetKwFunction{AssignLevels}{AssignLevels}
\SetKwFunction{ModifyPath}{ModifyPath}
\SetKwFunction{UpdateTree}{UpdateTree}
\KwIn{Spanning Tree $T$ rooted at $s$.}
\KwOut{A modified tree $\overline{T}$.}
\BlankLine

$\overline{T} \gets \phi$  \tcp*{$T = P = \{P_{\kappa-1}, P_{\kappa-2},\dots,P_1, P_0 \}$}

\tcp{Assign Levels to all nodes in all regular paths in $T$.}
$i \gets \kappa-1$  \tcp*{start from second level from top}

\While{$i \geq 0$}{
	\ForEach{$p_i \in P^{reg}_i$}{
		\tcp{$v_a$ and $v_b$ are the start and end nodes of path $p_i$}
		$H \gets \{v_a, v_b\}$\ \tcp*{$v_a$ is at same level as that of $p_i$.}
		\AssignLevels($p_i$, $H$, $i$)\;
		$\overline{T} \gets \overline{T} \cup p_i$\;
	}
	$i \gets i-1$\; 
}

\tcp{Pruned paths in $\overline{T}$ - Modify paths and assign levels.}

$\omega \gets \kappa-2$\;
\While{$\omega > 0$}{
	\ForEach{$p_{\omega} \in P^{pr}_{\omega}$}{
		$\overline{p}_{\omega} \gets \ModifyPath(p_{\omega},p_i)$ \tcp*{$p_\omega$ intersects $p_i$, $i > \omega$ and $v_b^\prime$ be the elected pseudo-leader. $p_i$ may be a modified path itself.}
		$\overline{T} \gets \overline{T} \cup \overline{p}_{\omega}$\;
		$H \gets \{v_a, v_b^\prime\}$ \tcp*{$v_a$ and $v_b^\prime$ are the start and end nodes of $\overline{p}_{\omega}$.}
		\AssignLevels($\overline{p}_{\omega}$, $H$, $\omega$)\;
	}
	$\omega \gets \omega-1$\;
}

\Return $\overline{T}$\;

\caption{Modified Tree}
\label{alg:ModifiedTree}
\end{algorithm}

Define {\ttfamily{AssignLevels}}$(p_i, H, i)$, where $H$ is a pair of end-nodes of $p_i$, to assign levels to all the nodes of $p_i$ by identifying maximal independent nodes (excluding the end nodes of $p_i$). Levels are assigned in the range $(i-1)$ to $0$. A modified path is connected to an alternate leader called \textit{pseudo-leader} by the function {\ttfamily{ModifyPath}}$(p_{\omega}, p_i)$ which chooses the nearest level-$(\omega+1)$ node on $p_i$ from the intersection point.

Consider that we are at some level $\omega$ where $0 \leq \omega \leq \kappa-1$ and suppose that there are several pruned paths in $P_{\omega}$. Let $p_\omega \in P_\omega$ be one such path and let $y \in p_i$ be the intersection point, where $i > \omega$. A \emph{pseudo-leader}, $v_{\omega+1}$, is chosen on $p_i$ using {\ttfamily{ModifyPath}} $(p_{\omega},p_i)$ (see appendix). Note that this may alter $I_i$ to $\overline{I_i}$ by replacing the original leader by the pseudo-leader. The path $p_\omega$ is extended from $y$ to $v_{\omega+1}$ and this new extended path $\overline{p}_{\omega}$ replaces $p_{\omega}$ in the modified tree $\overline{T}$. Once a new path $\overline{p}_{\omega}$ is established, all the nodes in it are assigned levels using ({\ttfamily{AssignLevels}}$(\overline{p}_{\omega}, H, \omega)$, where $H$ is the set of end-nodes of $\overline{p}_{\omega}$). This procedure of modifying pruned paths, replacing the old pruned paths by new, extended, \emph{modified} paths and assigning levels to all nodes in those paths is repeated for all levels down to $0$. The resulting tree is a modified tree with normal leaders and pseudo-leaders for respective types of paths.

%% file: Analysis.tex
\section{Analysis}
\label{analysis}

We will analyze the algorithm based on the data fusion algorithm in the modified tree $\overline{T}$. The fusion algorithm works in phases, where each phase consists of $\kappa$ rounds. In the beginning of each phase, several source nodes may have data to send. In the first round, each source node sends its data to its respective parent at level 1. In the second round, the leaders at level 1 send the fused data to their respective parent at level 2. In general, in round $i$, where $1 \leq i \leq \kappa$, the leaders at level $i-1$ send their data to their respective parent at level $i$ which then fuses the received data. At the end of round $\kappa$, the sink node $s$ would have fused all the received data.

The modified tree $\overline{T}$ naturally defines a hierarchical partition of $G$. For each node $u \in \overline{I_i}$, we define the respective cluster $Z_i^u$ to contain all the nodes in $V$ which appear in the subtree of $\overline{T}$ rooted at $u$ at level $i$. Node $u$ is the {\em leader} of $Z_i^u$. We denote by $Z_i = \bigcup_{u \in \overline{I_i}} Z_i^u$ the partition of $G$ at level $i$. Let $Z = \bigcup_{0 \leq i \leq \kappa} Z_i$ denote a {\em partition hierarchy} of $V$ with $\kappa$ levels. Note that $Z_\kappa$ consists of one cluster containing all the nodes in $V$. Clearly, each cluster $Z_i^u$ induces a connected subgraph in $G$. 

Let $A$ denote all the source nodes. We assume that each data item generated at a source node has size 1. Let $A_i^u = A \cap Z_i^u$. According to the data fusion algorithm, at the end of a round no later than $i$, the data items in $A_i^u$ will appear fused at node $u$ with size $f(|A_i^u|)$. Let $B_i^j$ denote the set of nodes at level $i$ which at the end of round $i$ hold data with size $f(|A_i^u|) \in [2^j,2^{j+1}-1]$, where $0 \leq j \leq \lambda$, and $\lambda = \lceil \log |A| \rceil = O(\log n)$. Let $B_0^0 = A$. Lemmas (\ref{theorem:upper-bound} and \ref{theorem:lower-simple}) establishes the lower-bound on communication cost at each round. 

A path $p_i^j$ could be intersected by multiple lower-level paths. Even though the leaders at a level $i$ are sufficiently far off, due to intersection by other paths, the leader at level $i$ might be close to many leaders of lower level paths. However, the number of such leaders that are close is \textit{limited}. Lemmas \ref{lemma:max-path-segments}, \ref{lemma:max-modified-paths} and \ref{lemma:max-pseudo-leaders} establishes the maximum number of pseudo-leaders in a given neighborhood. If all the nodes of a path $p_i$ belong to a cluster $Z_i^u$, we call such a path to be a \textit{total internal path}. Let $\delta = 3\cdot2^i$.

\begin{lemma}
\label{lemma:internal-path-len}
The maximum distance in $G$ between any node $v \in Z_i^u$ to $u$ is $\delta = 3\cdot2^i$ and there is a total internal path from any node $v \in Z_i^u$ to $u$ with respect to $Z_i^u$.
\end{lemma}
\begin{proof}
	Consider a path $p_{v,u} \in Z_i^u$. In the worst case, this path could be a concatenation of several modified paths, ranging from level $0$ to $i-1$. The total length of $p_{v,u}$ would be equal to the sum of maximum lengths of each of those segments: $\sum_{j=0}^{i-1}(3\cdot2^j - 2) = 3\cdot2^i-2i-1 < 3\cdot2^i$. 

By construction, any cluster $Z_i^u$ will contain those nodes of $V$ that appear in the subtree of $\overline{T}$ rooted at $u$ at level $i$. The path $\overline{p_i} \in \overline{T}$ from any of the member nodes $v \in Z_i^u$ to $u$ contains \textit{only} nodes that lie within the cluster $Z_i^u$. Also, $\overline{p_i}$ translates to a path $p \in T$. 
\end{proof}

\begin{lemma}
\label{lemma:max-pseudo-leaders}
	The total number of pseudo-leaders at level $i$, where $i \leq r$, which are inside $N_G(x,2^r)$ is at most $2^{2\rho(r-i+3)}\cdot(\kappa-i+1)^2$.
\end{lemma}
\begin{proof}
	From Lemma \ref{lemma:max-path-segments}, there are $2^{\rho(r-i+3)} \cdot (\kappa-i+1)$ path segments $p_{i+j} \in T$, $j \geq 0$, crossing $N(x,2^r)$. From Lemma \ref{lemma:max-modified-paths}, each such path segment can have multiple modified path segments at level $i$ or higher passing through it $(\leq 2^{\rho(r-i+1)} \cdot (\kappa-i+1))$, the total number of modified path segments that cross $N(x, 2^r)$ would be atmost $2^{2\rho(r-i+3)}\cdot(\kappa-i+1)^2$. This gives also an upper bound to the number of pseudo-leaders at level $i$ or higher.
\end{proof}

\begin{lemma}
\label{theorem:upper-bound}
The communication cost of the data fusion algorithm at round $i$ is at most $|B_{i-1}^j| 6 \cdot 2^{i+j}$, for any $1 \leq i \leq \kappa$.
\end{lemma}

\begin{proof}
Let $v$ be a leader in $\overline{I}_{i-1}$ and suppose that $v \in B_{i-1}^j$. The cost of sending data from $v$ to its parent at level $i$ is at most $\delta\cdot(2^{j+1}-1) \leq 3 \cdot 2^{i + j + 1}$, which is obtained by the product of the total path length connecting $v$ to its parent (bounded by $\delta$), from Lemma \ref{lemma:internal-path-len}, to the size of the data (bounded by $2^{j+1}-1$). Therefore, the total cost at round $r$ for communicating data in set $B_{i-1}^j$ is at most $|B_{i-1}^j| 6 \cdot 2^{i + j}$.
\end{proof}

We proceed with estimating a lower bound on the communication cost at each round $k$. Let $C^*(A,G)$ denote the optimal communication cost for $A$ data sources in graph $G$. We start with a simple lower bound.

\begin{lemma}
\label{theorem:lower-simple}
If a message is sent in round $i$ from a node $v \in B_{i-1}^j$ then $C^*(A,G) > \max(2^{i +j -1},1)$.
\end{lemma}

\begin{proof}
Since $v \in B_{i-1}^j$, the message sent by $v$ during round $i$
has size at least $2^{j}$. If $i = 1$ then the lower bound $1$ of the claim follows trivially. Suppose now that a message is sent during round $i > 1$. Let $\Gamma \subseteq A$ be the set of source nodes such that for every $x \in \Gamma$, $\dist_{G}(x,s) > 2^{i-1}$. We will show that $f(|\Gamma|) \geq 2^{j}$. Assume for contradiction that $f(|\Gamma|) < 2^j$. From the construction of $\overline{T}$ it can be verified that every node within distance $2^{i-1}$ from $s$ is a member of the subtree of $\overline{T}$ rooted at $s$ at level $i-1$, since $s$ is always preferred as a leader choice; these are the nodes of set $A \setminus \Gamma$. Thus, the last round of the algorithm in which a message is sent
carrying information from data sources in $A \setminus \Gamma$ is $i-1$.
Therefore, the message which is sent at round $i$ must carry information only from the data sources in $\Gamma$. Therefore, $f(|\Gamma|) \geq 2^{j}$. The smallest cost to transfer the data from the $\Gamma$ data sources to the sink $s$ is more than $f(|\Gamma|) 2^{i-1} \geq 2^j 2^{i-1} = 2^{i+j-1}$.
\end{proof}


We continue with an alternative lower bound. We construct a new family of multi-graphs by contracting the clusters of the partition hierarchy $Z$. Let $\G_i = (\V_i, \E_i, \w)$ be complete graph that we obtain when we contract the edges in the clusters of $Z_i$ and the whole cluster $Z_i^u$ is replaced by $u$. Each edge $e = (x,y) \in E$ is replaced by edge $e' = (x',y') \in \E_i$ where $x'$ and $y'$ are the respective leaders of $x$ and $y$ at partition $Z_i$ and $\w(x',y')= dist_G(x',y')$.

\begin{lemma}
\label{theorem:subsets}
Given an arbitrary set of nodes $X \subseteq \V_i$, where $0 \leq i \leq \kappa$, there is a subset $X' \subseteq X$ such that $|X'| \geq |X|/(2^{10\rho} \cdot (\kappa-i+1)^2 + 1)$, and for each pair of distinct $x,y \in X'$, $dist_{\G_i}(x,y) > 3\cdot\delta$.
\end{lemma}

\begin{proof}
Construct a new unweighted graph $H = (X, E_H)$ such that $(x,y) \in E_H$ if $\dist_{\G_i}(x,y) \leq 3\cdot\delta$. Let $x \in X$. From Lemma \ref{lemma:max-pseudo-leaders}, the number of pseudo-leaders of level $i$ which are in $N_G(x,3\delta)$ is bounded by $m = 2^{2\rho(i+2-i+3)}\cdot(k-i+1)^2 = 2^{2\rho(5)}\cdot(k-i+1)^2 = 2^{10\rho} \cdot (\kappa-i+1)^2$. By construction of $\G_i$, $\dist_{\G_i}(x,y) \leq 3\delta$ only if $y \in N_G(x,3\delta)$. Therefore, the degree of $H$ is bounded by $m$. Hence, $H$ accepts a $m+1$ coloring which implies that there is an independent set $X' \subseteq X$ in $H$ of size $|X'| \geq |X|/(2^{10\rho} \cdot (\kappa-i+1)^2 + 1)$ . Clearly, for each $x,y \in X'$, $dist_{\G_i}(x,y) > 3\cdot\delta$.
\end{proof}

\begin{lemma}
\label{theorem:lower-bound}
$C^*(A,G) \geq 2^{i+j-1} \{|B_i^j|/[2^{10\rho} \cdot (\kappa-i+1)^2 + 1] - 1\}$, for every $i$, $0 \leq i \leq \kappa$., $0 \leq j \leq \lambda$.
\end{lemma}

\begin{proof}
Consider a particular $B_i^j$. Let $\mu = |B_i^j|$ and $B_i^j = \{ u_1, \ldots, u_{\mu} \}$. We have that $A_i^{u_q}$ data sources are in $Z_i^{u_q}$ and $f(|A_i^{u_q}|) \in [2^j, 2^{j+1}-1]$. Let $X = \bigcup_{q = 1}^{\mu} A_i^{u_q}$. Let $C^*(X,G)$ denote the optimal
cost of fusing the data from the $X$ data sources to the sink node $s$ in graph $G$.

We transform the data fusion problem with the $X$ sources to a new data fusion problem in $\G_i$ where all data sources in $A_i^{u_q}$ are relocated to the respective leader node $u_q$.
The size of the data held at the leader node is $f(|A_i^{u_q}|)$. Let $\X = B_i^j$. We have that the optimal cost of fusing the data from $\X$ in $\G$ is no larger than the cost of fusing the data from $X$ in $G$, namely, $C^*(\X, \G_i) \leq C^*(X, G)$, since graph $\G_i$ is obtained from $G$ by contracting edges, and thus the optimal communication cost of $G$ cannot get worse in $\G$.

From Lemma \ref{theorem:subsets}, there is a subset $\X' \subseteq \X$, such that $|\X'| \geq |\X|/[2^{10\rho} \cdot (\kappa-i+1)^2 + 1]$, and for each distinct $x,y \in \X'$, $\dist_{\G}(x,y) > 3\delta$. Let $x,y \in \X'$ be any two nodes $u \in Z_i^x$, $v \in Z_i^y$ and $\dist_G(x,y) \geq \delta$. The cost of fusing the data stored in $\X'$ to the sink $s$ is at least the weight of the Steiner tree $T_S$ in $G_i$ that connects the set of nodes $\X'$ (including $s$ in the Steiner tree would only increase its cost) multiplied by $2^{j}$ which is a lower bound on the size of the data from each source in $\X'$. The total weight of $T_S$ is
$w(T_S) > \delta (|\X'|-1) \geq \delta \{|B_i^j|/[2^{10\rho} \cdot (\kappa-i+1)^2 + 1] - 1\}$.
The claim follows, since $w(T_S) \cdot 2^{j} \leq C^*(\X', \G_i) \leq C^*(\X, \G_i) \leq C^*(X, G) \leq C^*(A, G)$.
\end{proof}

\begin{theorem}[Approximation of Spanning Tree]
\label{theorem:fusion-approx}
Tree $\overline{T}$ is $\sigma$-OST where $\sigma = O(2^{10\rho} \log^3 D \cdot \log n)$.
\end{theorem}

\begin{proof}
Suppose that there is at least one data source different than the sink $s$. The algorithm runs in $\kappa = O(\log n)$ rounds. Let $Q_{i-1}^j$ be the cost of the algorithm when transferring data from  the level $i-1$ nodes $B_{i-1}^j$ to the respective parent leader in level $i$. We have that the total cost of the algorithm at round $i$ is $Q_i = \sum_{j = 0}^{\lambda} Q_{i-1}^j$. From Lemma \ref{theorem:upper-bound},
\begin{equation}
\label{eqn:1}
Q_{i-1}^j \leq |B_{i-1}^j| 6 \cdot 2^{i+j}.
\end{equation}
Let $Q_{i'-1}^{j'} = \max_{i,j}Q_{i-1}^j$. Since at least one message is sent in the algorithm, $Q_{i'-1}^{j'} \geq 1$. From Lemma \ref{theorem:lower-simple},
\begin{equation}
\label{eqn:2}
C^*(A, G) > \max(2^{i'+j'-1},1).
\end{equation}
From Lemma \ref{theorem:lower-bound}, we also have
\begin{equation}
\label{eqn:3}
C^*(A, G) \geq 2^{i' + j' - 1} \left (\frac {|B_{i'-1}^{j'}|} {2^{10\rho} \cdot (\kappa-i+1)^2 + 1} - 1 \right).
\end{equation}

If $|B_{i' - 1}^{j'}| < 2(2^{10\rho} \cdot (\kappa-i+1)^2 + 1)$, then, from Equations \ref{eqn:1} and \ref{eqn:2}
we have:

\begin{equation}
\label{eqn:4}
\frac{Q_{i'-1}^{j'}}{C^*(A,G)} < \frac{2(2^{10\rho} \cdot (\kappa-i+1)^2 + 1) \cdot 6 \cdot 2^{i'+j'}}{2^{i'+j'-1}}
				 \leq 24 \cdot \left (2^{10\rho} \cdot (\kappa+1)^2 + 1 \right)
\end{equation}

If $|B_{i'-1}^{j'}| \geq 2(2^{10\rho} \cdot (\kappa-i+1)^2 + 1)$, then, from (\ref{eqn:1}) and (\ref{eqn:3}) we also obtain (\ref{eqn:4}). Let $Q $ be the total cost of the algorithm, $Q = \sum_i Q_i$. Since there are $\kappa$ levels and $\lambda$ clusters, from Equation \ref{eqn:4} we have
$Q \leq \kappa (\lambda+1) Q_{i'-1}^{j'} \leq \kappa (\lambda + 1) \cdot C^*(A,G) \cdot 24 \cdot \left (2^{10\rho} \cdot (\kappa+1)^2 + 1 \right).$
Thus, since $\kappa = O(\log D)$ and $\lambda = O(\log n)$, we have
$\frac {Q} {C^*(A,G)} \leq \kappa (\lambda + 1) \cdot 24 \cdot \left (2^{10\rho} \cdot (\kappa+1)^2 + 1 \right) = O(2^{10\rho} \log^3 D\cdot\log n).$
\end{proof}

If $f$ is the constant cost function, then, $\lambda = 0$ and we obtain the following corollary.

\begin{corollary}
The constant cost function $f(x) = c$,  gives $O(2^{10\rho}\log^3D)$-oblivious Steiner tree from Theorem \ref{theorem:fusion-approx}.
\end{corollary} 

%% file: Simulation.tex
\section{Simulation Results}

We simulated the Oblivious Spanning Tree and compared its performance (fusion cost) with GRID\_GIST \citep{JiaNRS06}, Maximum Matching Algorithm \citep{644191} and other common trees such as MST (Minimum Spanning Tree) and SPT (Shortest-Paths Tree). We used a 2-D grid topology for our simulation using NetworkX \citep{networkx}. 2-D grids are a special case of doubling dimension graphs and they fall under a variation of the Steiner tree problem called `Rectilinear Steiner Problem' (RSP) where the tree structure has only vertical and horizontal lines that interconnects all points and is proved to be NP-Complete \citep{RectilinearSteiner}. Since calculating a minimum weight tree structure in a 2-D grid topology (a doubling-dimension graph) is essentially an RSP, the problem we are addressing is NP-Hard. 

We build a single spanning tree in a grid with 1600 nodes. We simulate it for random sets of data sources, upto 1445, that are randomly placed. Note that GRID\_GIST is a special algorithm designed for grids and ours is a generalized algorithm. Hence, GRID\_GIST performs slightly better than OST (in Fig~\ref{fusion} in appendix). 

%% file: Conclusion.tex
\section{Conclusions and Future Work}

We studied the problem of computing near-optimal fusion trees when the number of sources and the fusion cost function are unknown in the context of \textit{Doubling Dimension Graphs}. We have demonstrated that a simple, deterministic, polynomial-time algorithm based on independent sets can provide a near-optimal data structure for data-fusion (under the assumptions of concave cost function and doubling dimension graphs). The algorithm constructs a spanning tree which is then altered to produce a modified tree. We have shown that this algorithm guarantees $\log^4 n$-approximation over the optimal cost. As part of our future work, we are looking into the same problem on planar graphs, extending to multiple sinks and incorporating fault-tolerance and load-balancing.

%% file: Proofs.tex
\section{Observations and Proofs}

\begin{lemma} [\textbf{Presence of a Pseudo-Leader}]
\label{lemma:pseudo-leader}
	The {\ttfamily{ModifyPath}}$(p_\omega, p_i)$ function gurantees selection of a $(\omega + 1)$-level pseudo-leader.
\end{lemma}

\begin{proof}
	Suppose path $p_w$ intersects a higher-level path $p_i$. Let the start-node of $p_\omega$ be $u$ and let the end-node of $p_i$ be $v$. Note that a path $p_i$ goes from level $i$ to level $i+1$. There could be two cases for the presence of a pseudo-leader in $p_i$. If level of $v$ is $\omega + 1$, then, $v$ itself acts as a pseudo-leader for $u$. If level of $v$ is greater than $\omega + 1$, then, $p_i$ must have some nodes (within its end-nodes) that have been assigned to level $\omega+1$ (by the {\ttfamily{AssignLevels}} function) . Hence, in either case, a psuedo-leader is guaranteed to be found in $p_i$ for $u$.
\end{proof}

\begin{figure}[!h]
\centering
\subfloat[Intersecting paths ]{\label{fig:prune1}
\includegraphics[scale=0.28]{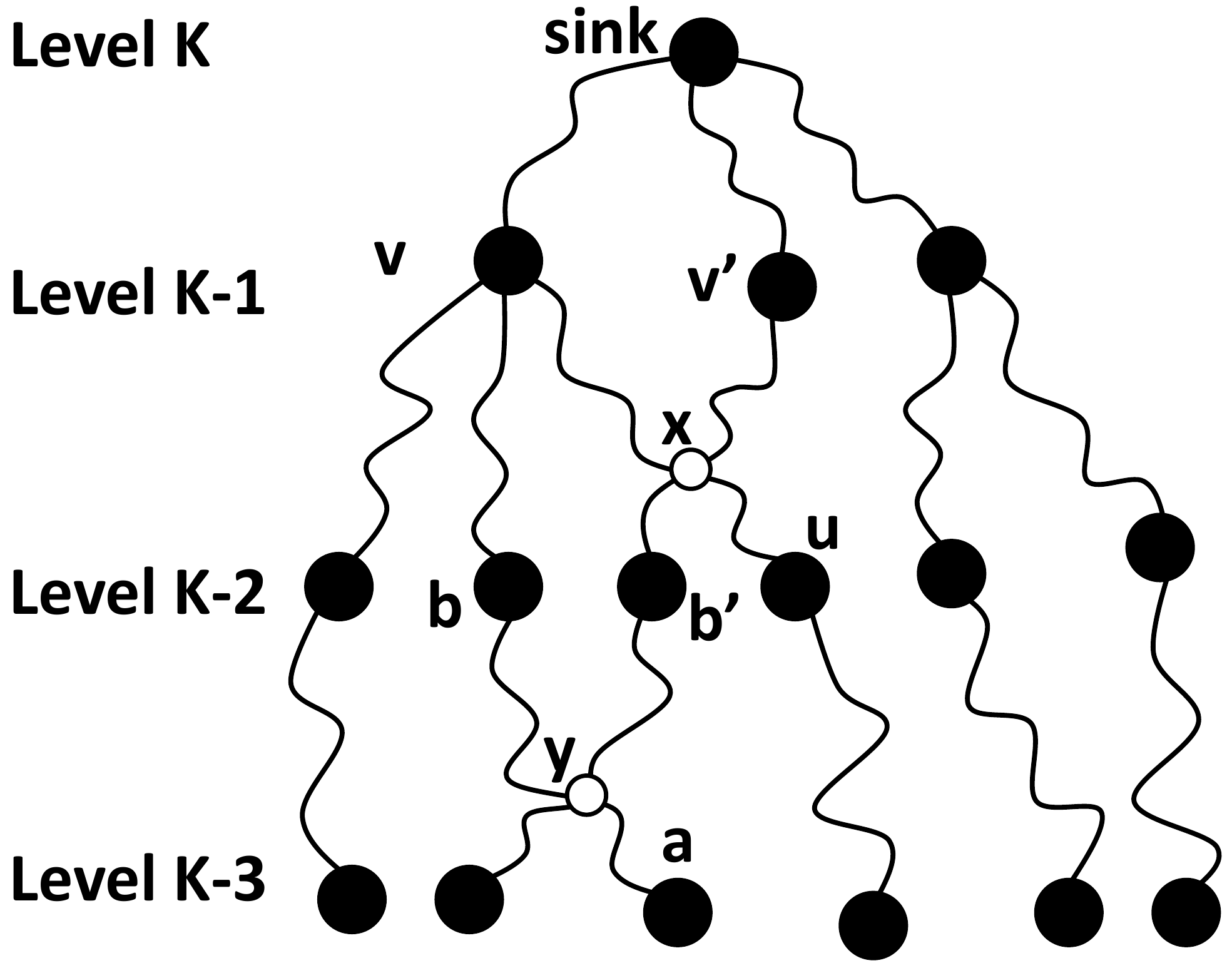}} \qquad
\subfloat[Pruned paths $y,b$ and $x,v^\prime$.]{\label{fig:prune2} 
\includegraphics[scale=0.28]{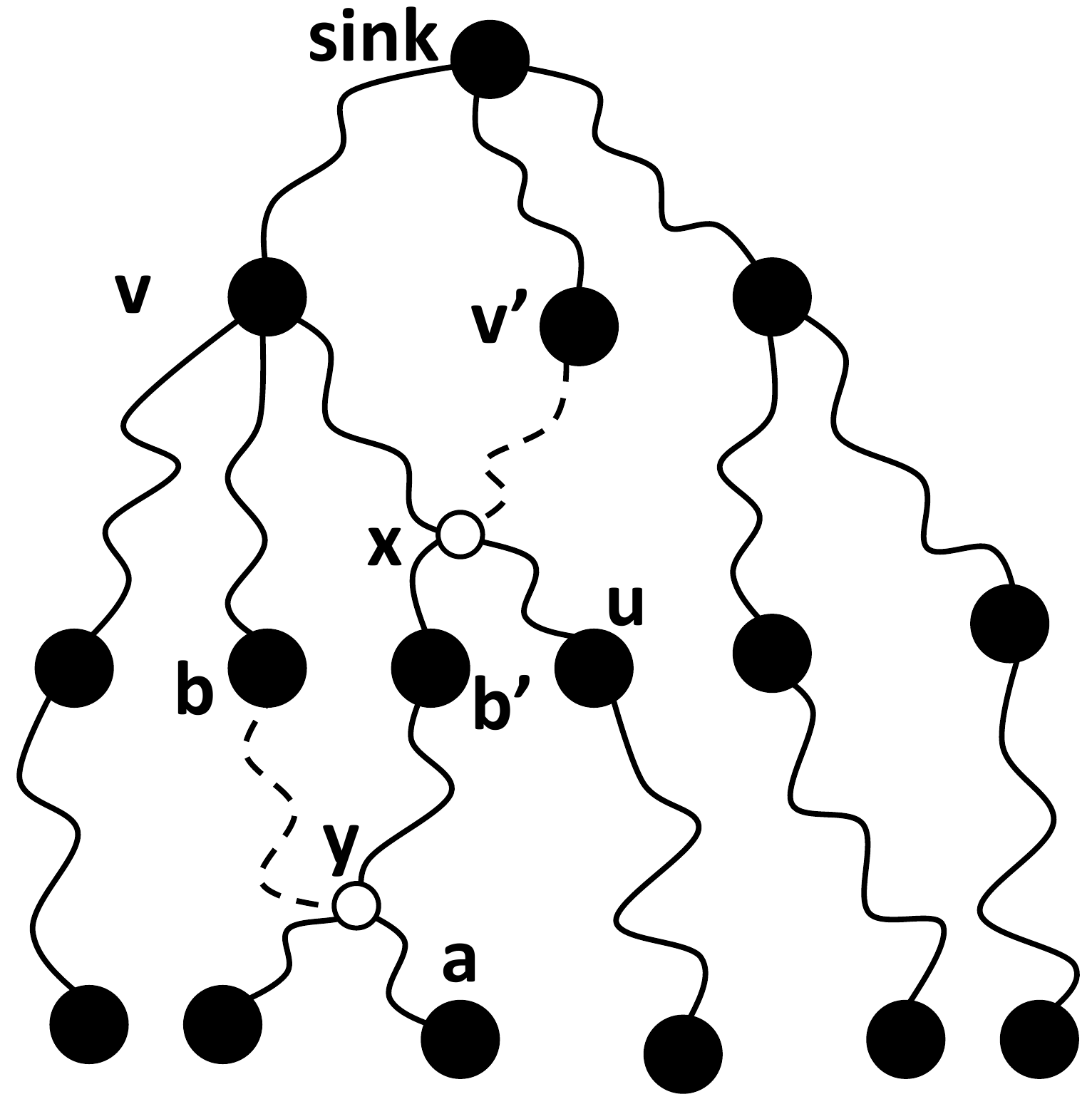}} \qquad
\subfloat[Modified Paths $p_{a,b^\prime}$ and $p_{b^\prime, v}$] {\label{fig:modified}
\includegraphics[scale=0.28]{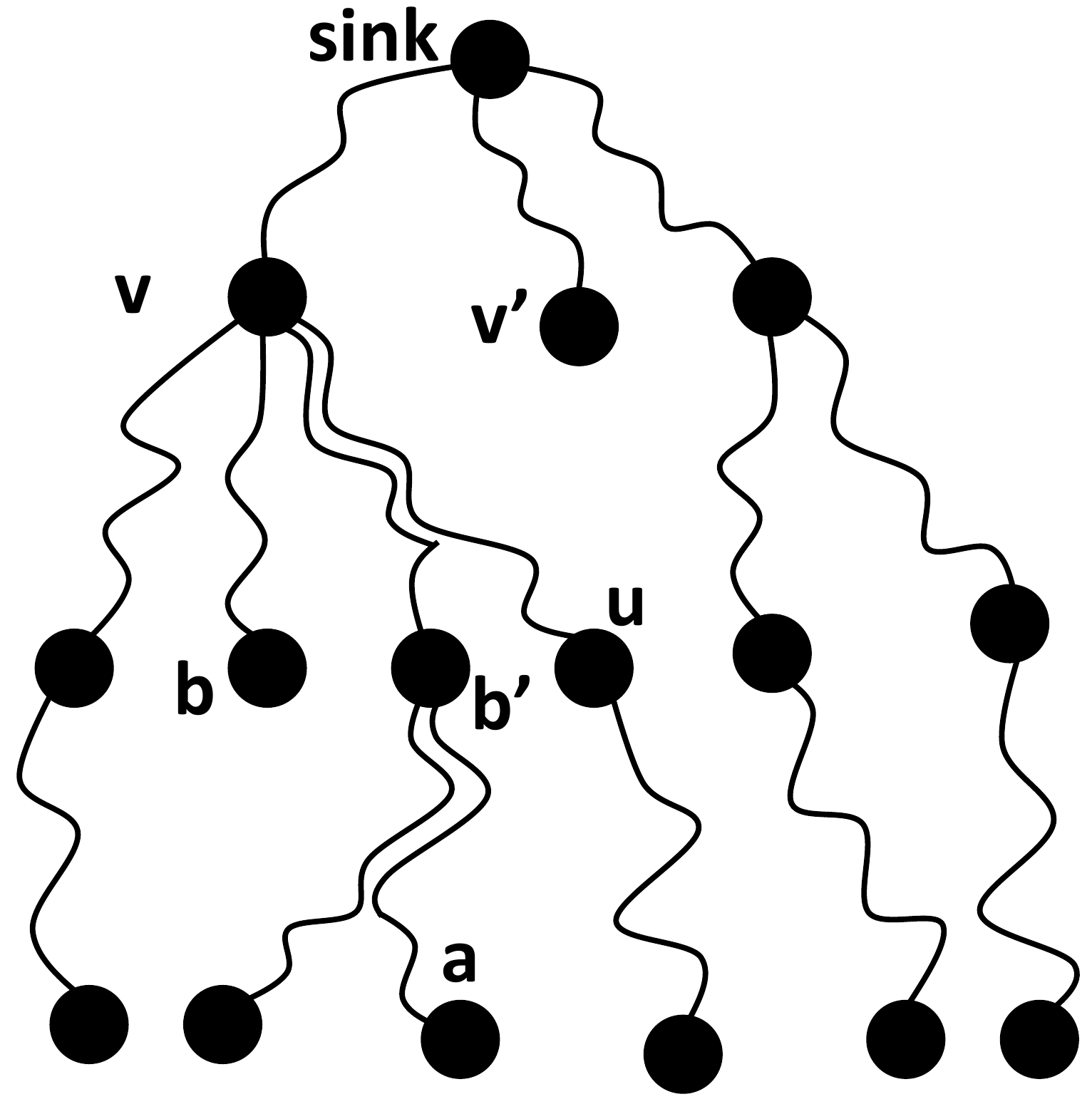}}
\caption{Pruning and Tree Modification.}
\label{fig:ModifiedPath}
\end{figure}

\begin{lemma}[\textbf{Upper Bound on $|p_{\omega}^{\prime}|$}]
\label{lemma:max-path-length}
	The upper bound on the length of $p_{\omega}^{\prime}$ is $(3\cdot2^{\omega} - 2)$.
\end{lemma}

\begin{proof}
   Consider a path $p_\omega$ that starts at $x \notin p_i$ and intersects another path $p_i$ at $y \in p_i$. Since $p_\omega$ is a pruned path, its distance from $x$ to the intersection point $y$ is atmost $2^\omega -1$ (if it was $2^\omega$ or more, point $y$ would have been its original leader). {\ttfamily {ModifyPath}} will attempt to seek the nearest $(\omega+1)$-level node (pseudo-leader) on $p_i$ from $y$ (Lemma \ref{lemma:pseudo-leader}). Note that $y$ itself cannot be the pseudo-leader for $x$ because, if it was, then, $p_\omega$ would not have been a pruned path. The distance from $y$ to a pseudo-leader $v$ on $p_i$ would be atmost $2^{\omega+1}-1$ because if this distance was more than $2^{\omega+1}-1$, we would have found another pseudo-leader $v^\prime$ that is $2^{\omega+1}$ distance away from $v$ and closer to $y$. This is due to the presence of $(2^{\omega+1})$-independent set nodes on this path $p_i$ computed by {\ttfamily{AssignLevels}}. Note that $y$ cannot be an end-node of $p_i$ and $v$ could be one of the end-nodes of $p_i$. Hence, the length of $p_\omega^\prime$ could be atmost $2^\omega -1 + 2^{\omega+1} -1 = 3\cdot2^\omega - 2$. Note that $p_i$ itself could be a stretched pruned path and the upper bound holds irrespective of the length of $p_i$.
\end{proof}

Figure \ref{fig:ModifiedPath} gives an example of intersecting path and its modification to reach a pseudo-leader and form a modified path.

\begin{function}
\KwIn{Paths $p_i$ and $p_m$ where $p_m$ intersects $p_i$ and $i > m$}
\KwOut{A modified path $\overline{p}_m$.}
\BlankLine

	\tcp{Let $p_m$ start from $x \notin p_i$ and intersect at $y \in p_i$ along 
	its path to leader $\ell_{m+1}$.}

	$v_{m+1} \gets \mbox{From $y$, identify the nearest level-$(m+1)$ node $v$} \in p_i$\;
	$p_m^a \gets \mbox{subpath from $x$ to $y$ in $p_m$}$\;
	$p_m^b \gets \mbox{subpath from $y$ to $v_{m+1}$ in $p_i$}$\;
	$\bar{p}_m \gets p_m^a + p_m^b$ \tcp*{Concatenate $p_m^a$ and $p_m^b$.}
	\Return $\overline{p}_m$\;

\caption{ModifyPath($p_m$, $p_i$)}
\label{fun:ModifyPath}
\end{function}


\begin{function}
\SetKwFunction{AssignLevels}{AssignLevels}
\KwIn{Path $p_i$, set of end-nodes $H$ of $p_i$ , level $i$.}
\KwOut{Assignment of levels to all nodes in $p_i$.}
\BlankLine

$L_{\lambda} \gets \phi$ \tcp*{Set of $2^{\lambda}$-independent nodes}

\For{$\lambda \gets (i-1)$ \KwTo $0$}{
	\tcp{Find $2^{\lambda}$-independent nodes at levels $\lambda = (i-1),(i-2),\ldots,1,0$.}
	$L_{\lambda} \gets MIS(p_i, H, 2^\lambda)$\;
	Assign level $\lambda$ to nodes in $L_{\lambda}$.
}

\caption{AssignLevels($p_i$, $H$, $i$)}
\label{fun:AssignLevels}
\end{function}


Consider an arbitrary node $x \in G$ with its neighborhood $N_G(x, 2^r)$, where $r \geq 0$.

\begin{lemma} [\textbf{Max path segments}]
\label{lemma:max-path-segments}
	The total number of path segments $p \in T$ at level $i$ or higher that cross $N_G(x,2^r)$ is at most $2^{\rho(r-i+3)} \cdot (\kappa-i+1)$.
\end{lemma}

\begin{proof}
	We know, by construction, that the length of a path $p_{i+j} \in T$ is atmost $2^{i+j}$ where $0 \leq j \leq (\kappa-i)$ and that there is atmost one leader $\ell_{i+j} \in I_i$ within $N(x, 2^{i+j}/2)$. Since we are looking at the number of path segments $p_{i+j}$ that go through $N(x,2^r)$, consider a large neighborhood $N(x,(2^{i+j} + 2^r))$ and determine the number of neighborhoods of radius $2^{i+j}/2$; $N(x,2^{i+j}/2)$. This is equivalent to atmost $2^{\rho{\log\frac{(2^{i+j} + 2^r)}{2^{i+j}/2}}} \leq 2^{\rho\log(r-i+3)} = 2^{\rho(r-i+3)}$ because of doubling dimension. For all paths that span the levels from $i$ to $\kappa$, the total number of path segments that cross $N(x,2^{i+j}/2)$ is equal to $2^{\rho(r-i+3)} \cdot (\kappa-i+1)$. 
\end{proof}

\begin{lemma} [\textbf{Max modified paths in a path segment}]
\label{lemma:max-modified-paths}
Consider a path segment $p \in T$ that crosses  $N_G(x,2^r)$.
The total number of modified paths ${\overline p} \in {\overline T}$ at level $i$ or higher that use nodes in $p \cap N_G(x,2^r)$ is at most $2^{\rho(r-i+1)} \cdot (\kappa-i+1)$.
\end{lemma}

\begin{proof}
	Let $Q = p \cap N_G(x,2^r)$. From Lemma \ref{lemma:max-path-length}, we know that the maximum length of any modified path $\overline{p}_{i+j}$ would be $3 \cdot 2^{i+j} -2 < 4 \cdot 2^{i+j} = 2^{i+j+2}$. To find the total number of modified paths $\overline{p}_{i+j}$ that passes through $Q$, we consider a larger neighborhood $N(x,2^{i+j+2}+2^r)$ and find the number of $N(y,2^(\frac{i+j+2}{2}))$ that would cover the larger neighborhood. Since each $\overline{p}_{i+j}$ has start node in $I_{i+j}$ and by doubling property of the graph, it can be computed to $2^{\rho\log{\frac{2^{i+j+2} + 2^r}{2^{i+j+1}}}} \leq 2^{\rho(r-i+1)}$. Since $j \in [0,(\kappa-i)]$, the total number of paths would be $2^{\rho(r-i+1)} \cdot (\kappa-i+1)$.
\end{proof}

\begin{figure}
\begin{center}
\includegraphics [scale = 0.70]{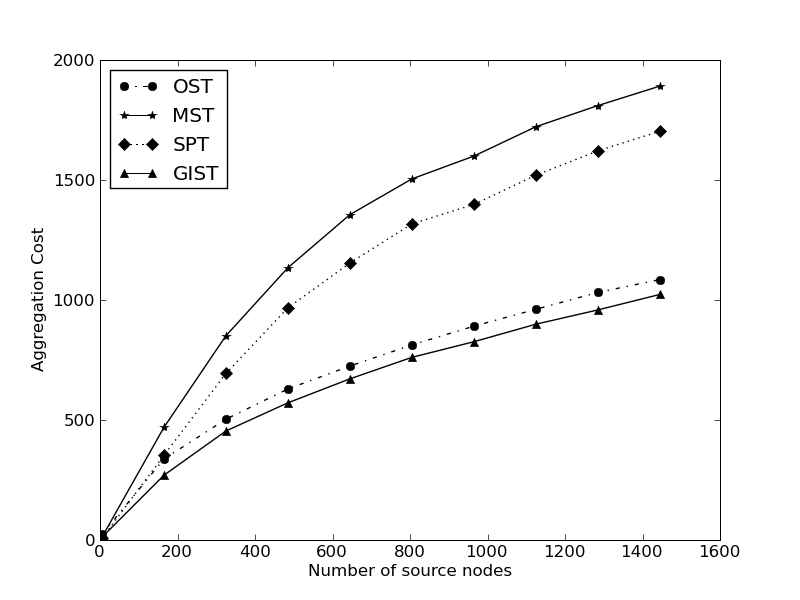}
\caption{Fusion Cost for varying set of source nodes in a 1600-node grid.}\label{fusion}
\end{center}
\end{figure}